\documentclass[11pt]{article}

\usepackage{fullpage}
\usepackage[utf8]{inputenc}
\usepackage[T1]{fontenc}

\usepackage{amsmath, amssymb, amsfonts}
\usepackage{amsthm}
\usepackage{graphicx}

\usepackage[colorlinks=true, allcolors=blue]{hyperref}

\usepackage{mathtools} 
\usepackage{array}     

\usepackage{subcaption}
\usepackage{float}
\usepackage{xcolor}
\usepackage{transparent}
\usepackage{graphicx}
\usepackage{enumitem}
\usepackage{orcidlink}

\graphicspath{{.}{Fig}}

\title{Quadratic polarity and polar Fenchel-Young divergences from the canonical Legendre polarity}

\author{
Frank Nielsen\thanks{Equal contribution.}~\orcidlink{0000-0001-5728-0726}\\
Sony Computer Science Laboratories Inc.\\
Tokyo, Japan
\and Basile Plus-Gourdon\footnotemark[1]\footnote{This work was carried out while Basile Plus-Gourdon was during his internship at National Institute for Informatics (NII), Tokyo, Japan}\\
\'Ecole Normale Sup\'erieure\\
Paris-Saclay, France
\and Mahito Sugiyama~\orcidlink{0000-0001-5907-9831}\\
National Institute for Informatics\\
Tokyo, Japan}

\date{}

\newtheorem{definition}{Definition}
\newtheorem{theorem}{Theorem}
\newtheorem{property}{Property}

\def\calT{\mathcal{T}}
\def\tD{{\mathrm{tD}}}
\def\bbP{{\mathbb{P}}}
\def\bbR{{\mathbb{R}}}

\def\GL{\mathrm{GL}}
\def\graph#1{{\mathrm{graph}(#1)}}
\def\epi#1{{\mathrm{epi}(#1)}}
\def\calL{\mathcal{L}}
\def\calM{\mathcal{M}}
\def\inner#1#2{{\langle #1, #2 \rangle}}
\def\natural#1#2{{(#1,#2)}}
\def\eqdef{{:=}}
\def\barbbR{\bar{\bbR}}
\def\st{\ :\ }

\begin{document}

\maketitle

\begin{abstract}
Polarity is a fundamental reciprocal duality of $n$-dimensional projective geometry which associates to points polar hyperplanes, and more generally $k$-dimensional convex bodies to polar $(n-1-k)$-dimensional convex bodies.
It is well-known that the Legendre-Fenchel transformation of functions can be interpreted from the polarity viewpoint of their graphs using an extra dimension.
In this paper, we first show that generic polarities induced by quadratic polarity functionals  can be expressed either as deformed Legendre polarity or as the Legendre polarity of deformed convex bodies, and be efficiently manipulated using linear algebra on $(n+2)\times (n+2)$ matrices operating on homogeneous coordinates.
Second, we define polar divergences using the Legendre polarity and show that they generalize the Fenchel-Young divergence or equivalent Bregman divergence.
This polarity study brings new understanding of the core reference duality in information geometry. 
Last, we show that the total Bregman divergences can be considered as a total polar Fenchel-Young divergence from which we newly exhibit the reference duality using dual polar conformal factors.
\end{abstract}

\noindent Keywords: Legendre transform; convex duality; projective geometry; polar; quadratic polarity; Fenchel-Young divergence; total Bregman divergence; optimal transport with quadratic cost.

\section{Introduction and contributions}

Let $\bbR^n$ denote the $n$-dimensional real vector space and $\bbR_n=(\bbR^n)^*$ its dual vector space of linear functions (covectors).

The Legendre-Fenchel transform~\cite{fenchel_conjugate_1949} is a fundamental conjugation operation in convex analysis. 
For a closed convex and proper function $F: \Theta \subset \bbR^n \to \barbbR$, its conjugate $F^*: H\subset\bbR^n \to \bar{\bbR}$ is defined as the supremum of linear functions $L_\theta(\eta):=\inner{\theta}{\eta}-F(\theta)$ where $\inner{\theta}{\eta}=\sum_{i=1}^n \theta_i\eta_i$ is the scalar product:

\begin{equation}\label{eq:LF_transform}
F^*(\eta) := (\calL F)(\eta) := \sup_{\theta \in \Theta} L_\theta(\eta).
\end{equation}

Since $F^*$ is the supremum of linear functions, it is necessarily convex and $\calL$ thus convexify functions.
Furthermore the biconjugation of a function $F$ lower bounds the function (i.e., ${F^*}^*\leq F$) with equality if and only if $F$ is convex and lower semi-continuous
 (Fenchel-Moreau's biconjugation theorem~\cite{borwein2006convex}).
When $F$ is of Legendre-type~\cite{rockafellar1967conjugates}, the convex conjugate is given by
$$
F^*(\eta) = \inner{(\nabla F)^{-1}(\eta)}{\eta}-F((\nabla F)^{-1}(\eta)),
$$
since the unique optimum value of Eq.~\ref{eq:LF_transform} is obtained for $\eta=\nabla F(\theta)$, or equivalently $\theta=\nabla F^{-1}(\eta)=\nabla F^*(\eta)$.
Historically, Legendre~\cite{Legendre1789} unravelled this dual parameterization $\theta\Leftrightarrow\theta$ in the calculus of variations in the 1D case.
Fenchel~\cite{fenchel_conjugate_1949} considered multivariate functions (not necessarily smooth) and relied on the Fenchel-Young inequality to define a wider scope of  convex duality using the notion of subdifferential and subgradient.
See~\cite{kiselman_werner_2019} for a survey of this transformation $\calL$.
The quadratic paraboloid function $Q(\theta)=\frac{1}{2}\sum_{i=1}^n \theta_i^2$ is the only function satisfying $Q=\calL Q=Q^*$~\cite{bauschke2012fenchel}.

The Legendre-Fenchel transformation serves as a cornerstone in diverse fields, ranging from classical mechanics where it facilitates the transition between Lagrangian and Hamiltonian formalisms~\cite{leok2017connecting}, to information geometry where it defines the dual coordinate systems of flat manifolds~\cite{amari_information_2016}.

The paper presents our contributions organized as follows: 

In \S\ref{sec:proggeo}, we recall the necessary background notions of (epi)graphs and projective geometry.
We then describe the Legendre-Fenchel transformation from the viewpoint of polarity in \S\ref{sec:LFTPolarity}:
Namely, we recall the seminal result of Werner Fenchel~\cite{fenchel_conjugate_1949} that the boundary of the Legendre polarity of the graph of a function  coincides with the graph of its convex conjugate (Proposition~\ref{prop:LegPolarity}).
In \S\ref{sec:polartransform}, we study several properties of generic quadratic polarities and exhibit two new identities of the polar Legendre transformations under transformations:
First, we show that an arbitrary quadratic polarity can be equivalently interpreted as convex body transformation $T$ 
of the Legendre polarity in Theorem~\ref{th:artstein_avidan_generalization_projective}.
Second,  we prove that a quadratic polarity is equivalent to the Legendre polarity after applying a deformation $S$ on the convex body in Theorem~\ref{th:nielsen_generalization_projective}.
The relationships between the transformations $T$ and $S$ are given in Proposition~\ref{prop:relationsTS}.
We define the Fenchel-Young polar divergence (Definition~\ref{def:breg_div}) in \S\ref{sec:polarfy} as a generalization of the Fenchel-Young divergence~\cite{acharyya2013learning} and recover its counterpart properties:
Namely, the non-negativeness and the duality under swapping in Proposition~\ref{prop:swap_bregman_arguments}.
In \S\ref{sec:polartbd}, we further define the total Fenchel-Young divergence (Definition~\ref{def:totalbreg_div}) which
is a generalization of the total Bregman divergence~\cite{vemuri_total_2011}.
We prove a duality identity under parameter swapping for the total Fenchel-Young divergence in Theorem~\ref{thm:swaptbD}.
Finally, we conclude with a discussion in \S\ref{sec:concl}.
Some of the proofs are deferred to the Appendix section in \S\ref{sec:appendix}.

\section{Background: (Epi)graphs, and projective geometry}\label{sec:proggeo}

Let $F:\Theta\subset\bbR^n \rightarrow \barbbR$ be an extended real-valued function.
The epigraph of $F$ is 
$$
\epi{F}=\{ (\theta,y)\in \Theta\times\bbR \st y\geq F(\theta) \}\subset\bbR^{n+1},
$$
and the graph of $F$ is $\graph{F}=\{(\theta,F(\theta)) \st \theta\in\Theta\}=\partial\epi{F}$, the boundary of the epigraph (where $\partial$ denotes the boundary operator).

In this work, we shall consider sets of $\bbR^{n+1}$ like the epigraphs of functions as basic objects of study.
Those sets are considered in the projective space~\cite{richter2011perspectives} $\bbP^{n+1}$ which extends $\bbR^{n+1}$ by adding ideal points at infinity.
$\bbP^{n+1}$ is defined as the quotient space $(\bbR^{n+2} \setminus \{0\})/\sim$ where the equivalence relation $\sim$ is defined as follows:
\begin{equation}\label{eq:equivalence_relation}
[a]\sim [b] \iff \exists \lambda \neq 0 \st [b]= \lambda\, [a].
\end{equation}
The vector components $a_1,\ldots, a_{n+2}$ of $[a]$ are called the homogeneous coordinates~\cite{nielsen2005visual}.
A vector $a$ of $\bbR^{n+1}$ can be considered as a vector $[a]$ of $\bbP^{n+1}$:
\begin{equation}\label{eq:relation_homogeneous_coordinates}
a=\begin{bmatrix} 
    a_1 \\ \vdots \\ a_{n+1} 
\end{bmatrix} \in\bbR^{n+1} \leftrightarrow [a] = \begin{bmatrix} 
    a_1 \\ \vdots \\ a_{n+1} \\ a_{n+2}=1 
\end{bmatrix} \in \bbP^{n+1}.
\end{equation}
Projective ideal points at infinity are characterized by $a_{n+2}=0$.
Conversely a non-ideal projective point $[a]$ (i.e., not lying at infinity) can be considered as an ordinary vector by dehomogeneization:
$$
[a] = \begin{bmatrix} 
    a_1 \\ \vdots \\ a_{n+2} \\ 1 
\end{bmatrix} \in \bbP^{n+1}, a_{n+2}\not=0 \Rightarrow  a=\begin{bmatrix} 
    \frac{a_1}{a_{n+2}} \\ \vdots \\ \frac{a_{n+1}}{a_{n+2}} 
\end{bmatrix} \in\bbR^{n+1}.
$$

Notice that graphs of $n$-variate functions in the projective space $\bbP^{n+1}$ are thus handled using homogeneous coordinates of $\bbR^{n+2}$ as depicted in Figure~\ref{fig:projective_space_homogeneous_coordinate_explained}.

\begin{figure} 
	\centering
  \def\svgwidth{0.5\columnwidth}%
  \graphicspath{{.}}%
\begingroup%
  \makeatletter%
  \providecommand\color[2][]{%
    \errmessage{(Inkscape) Color is used for the text in Inkscape, but the package 'color.sty' is not loaded}%
    \renewcommand\color[2][]{}%
  }%
  \providecommand\transparent[1]{%
    \errmessage{(Inkscape) Transparency is used (non-zero) for the text in Inkscape, but the package 'transparent.sty' is not loaded}%
    \renewcommand\transparent[1]{}%
  }%
  \providecommand\rotatebox[2]{#2}%
  \newcommand*\fsize{\dimexpr\f@size pt\relax}%
  \newcommand*\lineheight[1]{\fontsize{\fsize}{#1\fsize}\selectfont}%
  \ifx\svgwidth\undefined%
    \setlength{\unitlength}{595.27559055bp}%
    \ifx\svgscale\undefined%
      \relax%
    \else%
      \setlength{\unitlength}{\unitlength * \real{\svgscale}}%
    \fi%
  \else%
    \setlength{\unitlength}{\svgwidth}%
  \fi%
  \global\let\svgwidth\undefined%
  \global\let\svgscale\undefined%
  \makeatother%
  \begin{picture}(1,0.61904762)%
    \lineheight{1}%
    \setlength\tabcolsep{0pt}%
    \put(0,0){\includegraphics[width=\unitlength,page=1]{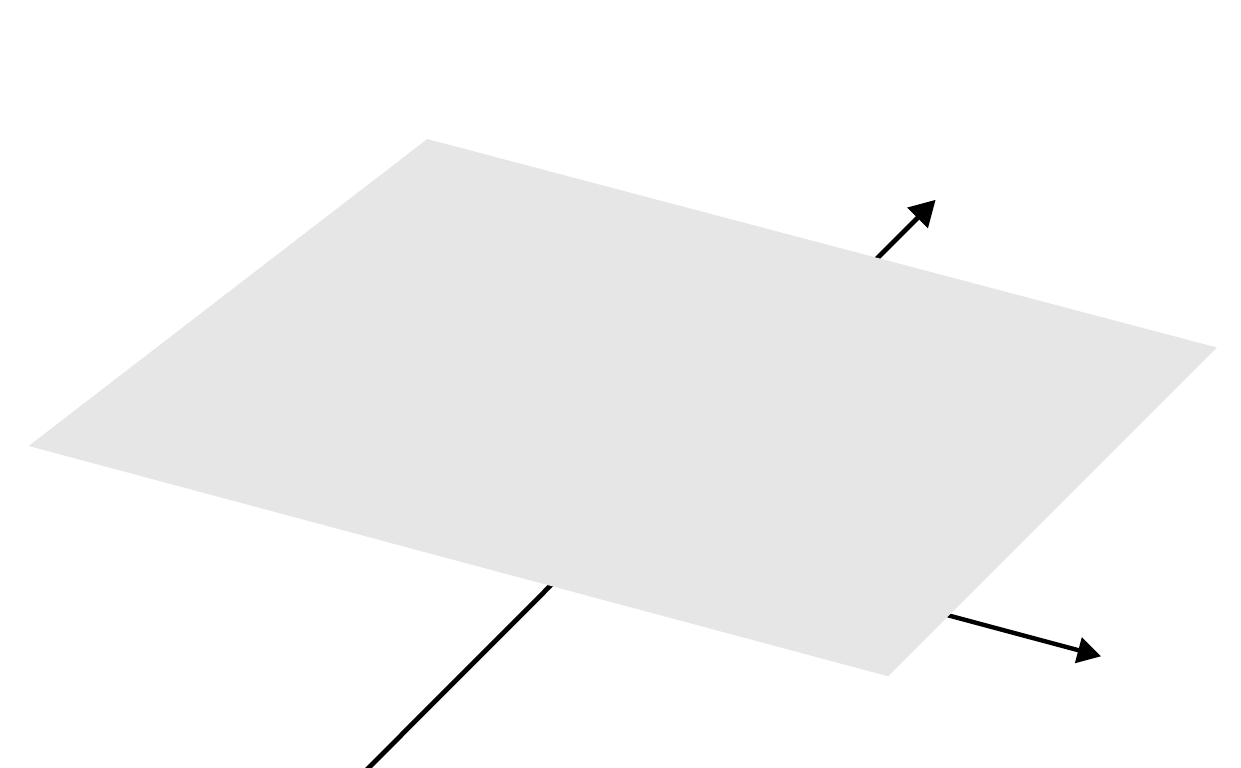}}%
    \put(0.26008688,0.37065346){\color[rgb]{0,0,0}\rotatebox{34.860153}{\makebox(0,0)[lt]{\lineheight{1.25}\smash{\begin{tabular}[t]{l}$z=1$\end{tabular}}}}}%
    \put(0.81082571,0.07861874){\color[rgb]{0,0,0}\rotatebox{-15}{\makebox(0,0)[lt]{\lineheight{1.25}\smash{\begin{tabular}[t]{l}$x$\end{tabular}}}}}%
    \put(0.75734732,0.41264007){\color[rgb]{0,0,0}\makebox(0,0)[lt]{\lineheight{1.25}\smash{\begin{tabular}[t]{l}$y$\end{tabular}}}}%
    \put(0.52202759,0.57109834){\color[rgb]{0,0,0}\makebox(0,0)[lt]{\lineheight{1.25}\smash{\begin{tabular}[t]{l}$z$\end{tabular}}}}%
    \put(0,0){\includegraphics[width=\unitlength,page=2]{projective_space_homogeneous_coordinate_explained.pdf}}%
    \put(0.55802616,0.25871735){\color[rgb]{0,0,0.66666667}\makebox(0,0)[lt]{\lineheight{1.25}\smash{\begin{tabular}[t]{l}$a=(x,F(x)) \in \mathbb R ^2$\end{tabular}}}}%
    \put(0.70249619,0.55399384){\color[rgb]{1,0.4,0}\makebox(0,0)[lt]{\lineheight{1.25}\smash{\begin{tabular}[t]{l}$[a]=[x,F(x),1] \in \mathbb P^2$\end{tabular}}}}%
  \end{picture}%
\endgroup%

		\caption{A $n$-variate function is represented by its $(n+1)$-dimensional epigraph (convex body) and manipulated using homogeneous coordinates of $\bbR^{n+2}$.}

		\label{fig:projective_space_homogeneous_coordinate_explained}
	\end{figure}

Homogeneous coordinates are useful to perform affine transformations and collineations (homographies) of projective vectors and often used in computer graphics and computer vision~\cite{nielsen2005visual}.
In projective geometry, there is a duality between points and lines (e.g., meet/join algebra~\cite{richter2011perspectives}).

\section{Legendre-Fenchel transformation from the viewpoint of polarity}\label{sec:LFTPolarity}

\subsection{Quadratic polarity and Legendre polarity}

In general, a polarity~\cite{hoffmann_recent_1988,Boroczky2008} $\Delta$ is a mapping  which satisfies the following property:
$$
\Delta(\cup_{i\in I} A_i) = \cap_{i\in I} \Delta(A_i),
$$ 
for any family $\{A_i\}_{i \in  I}$ such that $A_i \subseteq \bbP^{n+1}$.

In the affine space $\bbR^{n+1}$, a polarity can always be described using a polarity functional 
$p:\bbR^{n+1} \times  \bbR_{n+1} \to \barbbR$ (Theorem 1.1~\cite{hoffmann_recent_1988})  by
 
\begin{equation}\label{eq:polarity_from_functional}
\Delta(A) := \left\{ [b] \in  \bbR_{n+1}) \st \forall [a] \in A,\ p(a,b) \geq 0 \right\}.
\end{equation}

We consider polarities $\Delta_C$ induced by a cost matrix $C\in\GL(n+2)$ (not necessarily symmetric) such that
\begin{equation}\label{eq:functional_from_C}
p(a,b) \eqdef [a]^\top\, C\, [b].
\end{equation}

The polar of  a point $[a]$ is a halfspace: 
$$
H_{[a]}= \Delta_C([a]) = \{ [b] \in \bbR_{n+1} \, \mid\,  [a]^{\top}\, C\, [b] \geq 0\}.
$$
		
The polar of the graph of a function $F$ yields a set. 
When $F$ is closed, the polarity gives the same image as the epigraph of $F$: 
		$$
    \Delta_C(\graph{F}) = \Delta_C(\epi{F}).
    $$
		
Notice that a set can always be considered as the intersection of halfspaces:
    $$
    \Delta_C(A) = \bigcap_{[a] \in A} H_{[a]}.
    $$

We can similarly define the dual polarity $\Delta^*$ as 
\begin{equation}
\Delta^*(B):=\{[a] \in \bbR^{n+1} \mid \forall [b] \in B, [b]^\top\, C^\top\, [a] \geq 0 \}.
\end{equation}

\begin{property}[Polarity involution]
Let $\Delta$ be the polarity associated with a non-degenerate matrix 
$C \in \GL(n+2)$. Then for any closed convex set
$A \subset \bbR^{n+1}$ we have
$$
\Delta_C^*(\Delta_C(A)) = A.
$$
\end{property}

We refer to Appendix~\ref{sec:appendix} for a proof.

The Legendre-Fenchel transformation can be defined by the {\em Legendre polarity}~\cite{fenchel_conjugate_1949, hoffmann_recent_1988, martinez2015contribution,Edelsbrunner-2018} $\Delta_{\calL}$ where
\begin{equation}\label{eq:LF_cost_matrix}
C = C_{\calL} := \begin{bmatrix}
            -I_n & 0 & 0 \\
            0 & 0 & 1 \\
            0 & 1 & 0
        \end{bmatrix},
\end{equation}
where $I_n$ denotes the $n\times n$ identity matrix.
Note that since $C_{\calL}$ is symmetric, we get self-dual Legendre polarity: $\Delta_{\calL}=\Delta_{\calL}^*$.

In Appendix~\ref{sec:pcpolar}, we consider another kind of polarity which maps the  parabola to the unit sphere.

\subsection{Legendre-Fenchel transformation from Legendre polarity}

We now relate the Legendre polarity with the Legendre transformation (see~\cite{fenchel_conjugate_1949}):
Namely, the boundary of the Legendre polarity of the graph of a function $F$ coincides with the graph of its convex conjugate $F^*$:

\begin{property}[Legendre transformation from Legendre polarity]\label{prop:LegPolarity}
Let $F:\Theta \to  \bbR$ be a real-valued function and $F^{*}:H \to \bbR$ be its Legendre convex conjugate.
Then it holds that
\begin{equation}
\partial \Delta_{\calL}(\graph{F}) = \graph{F^{*}}.
\end{equation}
\end{property}

\begin{proof}
    Let $[b] = \begin{bmatrix} 
        \eta^\top & y^{*} & 1 
    \end{bmatrix}^\top  \in \Delta_{\calL}(\epi{F})$, then by definition of $\Delta_{\calL}$, 
		
    \begin{align*}
        \forall \theta \in  \Theta, p \left(\begin{bmatrix} 
        \theta^\top & F(\theta) & 1 
    \end{bmatrix}^\top, \begin{bmatrix} 
        \eta^\top & y^{*} & 1 
    \end{bmatrix}^\top  \right) \geq 0 & \iff
    \forall \theta \in  \Theta, \begin{bmatrix} 
        \theta^\top & F(\theta) & 1 
    \end{bmatrix} C_{\mathcal{L}} \begin{bmatrix} 
        \eta^\top & y^{*} & 1 
    \end{bmatrix}^\top \geq 0, \\
    & \iff \forall \theta \in  \Theta, y^{*} \geq \theta^\top \eta - F(\theta), \\
    & \iff y^{*} = \sup_{\theta \in \Theta}(\theta^\top \eta - F(\theta)):=F^*(\eta).
    \end{align*}

    The supremum exists and is reached because $F$ is closed and convex. 
		So $[b] \in  \Delta_{\calL}(\epi{F}) \iff [b] \in  \epi{F^*}$.
\end{proof}

    \begin{figure}
        \centering
        \includegraphics[width=0.5\textwidth]{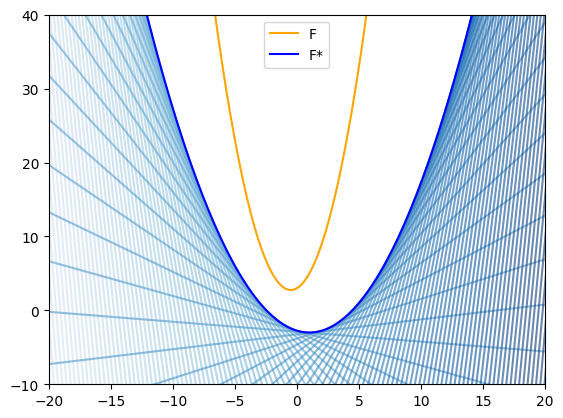}
				
        \caption{The graph of the Legendre transform $F^{*}$ of a function $F$ is the envelope of a family of polar hyperplanes called. 
				Each hyperplane in blue is associated to a point $(\theta,F(\theta)) \in  \graph{F}$.
				Here, we consider $F(\theta)=\theta^{2}+\theta+3$ with convex conjugate $F^*(\eta)=\frac{1}{4}(\eta^2-2\eta-13)$.}
        \label{fig:LF_as_envelope}
    \end{figure}

\subsection{Duality and properties of polarities}

		Each point $[b] \in  \bbR_{n+1}$ can be represented as one hyperplane $H_{[b]}$ in $\mathbb{R}^{n+1}$ given by 
\begin{equation}\label{eq:polar_to_b}
H_{[b]} = \{[a] \in  \mathbb{R}^{n+1} \mid [a]^\top C [b]=0\}.
\end{equation}

This hyperplane is called the polar to $[b]$. Similarly, the polar of a point $[a] \in \mathbb{R}^{n+1}$ is a hyperplane of $\bbR_{n+1})$.
\begin{equation}
    H_{[a]} = \{[b] \in \bbR_{n+1} \mid [b]^\top C^\top [a] =0 \}.
\end{equation}

The polar of a point $[a]$ is the hyperplane delimiting the half-space $\Delta([a])$.
Figure~\ref{fig:polar_and_polarity_duality_projective_space} illustrates the reciprocal duality of polarities.

\begin{figure}[H]
		\centering

  \def\svgwidth{0.5\columnwidth}%
  \graphicspath{{.}}%
  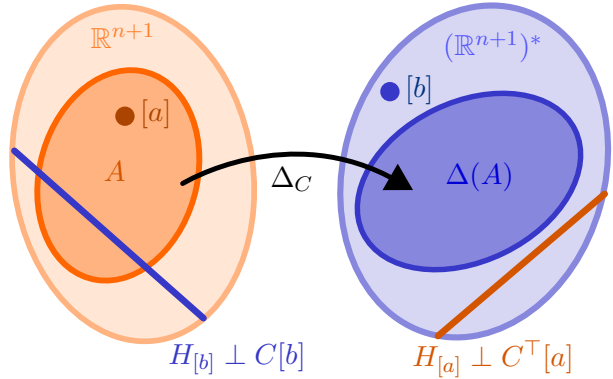%

				\caption{Illustration of the quadratic polarity $\Delta_C : \mathbb{R}^{n+1} \to \mathbb{R}_{n+1}$ induced by $(n+2)\times (n+2)$ matrix $C$:
				Polarity maps a set $A \subset \mathbb{R}^{n+1}$ to its dual  $\Delta_C(A) \subset \mathbb{R}_{n+1}$.
				 A point $[b]$ in $\mathbb{R}_{n+1}$ can be represented in $\mathbb{R}^{n+1}$  as a hyperplane $H_{[b]}$ of normal vector $C\, [b]$.
				}
				

		\label{fig:polar_and_polarity_duality_projective_space}
	\end{figure}

\begin{property}
    Let $A \subset \mathbb{R}^{n+1}$,
    \begin{equation}
    \forall [a] \in A, \quad \Delta(A) \subseteq \Delta([a]).
    \end{equation}
\end{property}

The proof is immediate given the definition of a polarity. Geometrically, it means that the dual set $\Delta(A)$ lies on the ``good side'' of the hyperplane \textit{polar} to $[a]$ and the polar of $[a]$ do not intersect the interior of $\Delta(A)$.

\begin{definition}
    [Supporting hyperplane]
    Given a normal vector $[n] \in  \mathbb{R}^{n+1}$, let $H$ be a hyperplane defined  by $H=\{[x] \in  \mathbb{R}^{n+1} \mid [x]^\top[n]=0\}$ and $A$ a properly convex set. $H$ is called \textit{supporting} hyperplane of $A$ at the point $[a_0] \in \partial A$ iff 
    \begin{itemize}
        \item (point $[a_0]$ belongs to the hyperplane) : $[a_0]^\top [n] = 0$
        \item (the set lies on one side of the plane) : $\forall [a] \in A$, $[a]^\top[n] $ is of constant sign
    \end{itemize}
\end{definition}

\begin{property}
    Let $A \subset \mathbb{R}^{n+1}$ be a closed convex set.
    \begin{equation}
    \forall [b] \in \partial \Delta(A), \quad H_{[b]} \text{ is a supporting hyperplane of } A.
    \end{equation}
\end{property}

\begin{proof}
    Let $[b] \in  \partial \Delta(A)$, by definition $H_{[b]} = \{[a] \in  \mathbb{R}^{n+1} \mid [a]^\top C [b]=0\}$. The normal vector of the hyperplane associated with $[b]$ is $C[b]$. As $[b] \in  \partial \Delta(A)$ there exists a $[a_0] \in \partial A$ such that $[a_0]^\top C [b]=0$. This gives the first condition for $H_{[b]}$ to be a supporting hyperplane. Next because $[b] \in  \Delta(A)$, we know that $\forall [a] \in  A, \quad [a]^\top C [b]\geq 0$. So $H_{[b]}$ is a supporting hyperplane of $A$ at point $[a_0]$    
\end{proof}

Geometrically, when the set $A$ is strictly convex and smooth, the boundary of the polar set $\Delta(A)$ is the envelope of the family of hyperplanes polar to each point $[a] \in \partial A$. Those extra assumptions are required to define the envelope of the supporting hyperplane : Smoothness guarantees that the sub-differential on the boundary is single valued and that we can parametrize the family of supporting hyperplanes while strict convexity ensures that the dual set $\Delta(A)$ is smooth avoiding a singularity on the dual boundary.

\begin{property}[Boundary of polar set is the envelope of polars] \label{prop:boundary_is_envelope_sets}
    Let $A \subset \mathbb{R}^{n+1}$ be a closed and convex set with smooth boundary $\partial A$. Let $a: \Theta \to \partial A$ be a parametrization of the boundary.
    
    The boundary of the polar set $\Delta_C(A)$ is exactly the envelope of the hyperplanes polar to the points of $\partial A$:
    \begin{equation}\label{eq:boundary_is_envelope}
    \partial \Delta_C(A) = \left\{ [b] \in \mathbb{R}_{n+1}) \st \exists \theta \in \Theta, \begin{cases} 
    [a(\theta)]^\top C [b] = 0 & \text{(Incidence)} \\
    \nabla_\theta [a(\theta)]^\top C [b] = 0 & \text{(Tangency)}
    \end{cases} \right\}
    \end{equation}
\end{property}

\begin{proof}
    Recall that the polar set of $A$ is defined as the intersection of half-spaces defined by points in $A$ (or equivalently its boundary $\partial A$):
    \begin{equation*}
    \Delta_C(A) = \left\{ [b] \in (\mathbb{R}^{n+1})^{*} \mid \forall [a] \in A, [a]^\top C [b] \geq 0 \right\}
    \end{equation*}
    Let $E$ denote the envelope set defined in the property. We prove the equality by double inclusion.

    \paragraph{1. $\partial \Delta_C(A) \subseteq E$}
    Let $[b] \in \partial \Delta_C(A)$. By definition of the boundary of a convex set, the inequality constraint must be active (tight) for at least one point in the primal set. Thus, there exists a parameter $\theta^* \in \Theta$ such that:
    \[
    f(\theta^*) := [a(\theta^*)]^\top C [b] = 0
    \]
    Since $[b] \in \Delta_C(A)$, we must have $f(\theta) = [a(\theta)]^\top C [b] \geq 0$ for all $\theta \in \Theta$.
    Consequently, $\theta^*$ is a global minimum of the smooth function $f(\theta)$. The gradient with respect to $\theta$ must vanish at this optimum:
    \[
    \nabla_\theta f(\theta^*) = \nabla_\theta [a(\theta^*)]^\top C [b] = 0
    \]
    The point $[b]$ satisfies both the incidence condition ($f(\theta^*)=0$) and the tangency condition ($\nabla f(\theta^*)=0$). Therefore, $[b] \in E$.

    \paragraph{2. $E \subseteq \partial \Delta_C(A)$}
    Let $[b] \in E$. By definition of the envelope, there exists a parameter $\theta_0$ such that:
    \begin{enumerate}
        \item $a(\theta_0)^\top C [b] = 0$
        \item $\nabla_\theta a(\theta_0)^\top C [b] = 0$
    \end{enumerate}
    To show $[b] \in \partial \Delta_C(A)$, we must prove two things:
    \begin{itemize}
        \item[(i)] $[b] \in \Delta_C(A)$ (Global validity): $[a(\theta)]^\top C [b] \geq 0$ for all $\theta$.
        \item[(ii)] $[b]$ is on the boundary: This is immediate since $a(\theta_0)^\top C [b] = 0$.
    \end{itemize}
    
    The tangency condition (2) implies that the supporting hyperplane $H_{[b]} = \{ [x] \mid [x]^\top C [b] = 0 \}$ is a tangent hyperplane to the boundary $\partial A$ at the point $a(\theta_0)$.
    
    Since $A$ is assumed to be convex, the set lies entirely on one side of any tangent hyperplane.
    This implies that $\theta_0$ corresponds to a global minimum of the scalar product, so:
    \[
    [a(\theta)]^\top C [b] \geq a(\theta_0)^\top C [b] = 0, \quad \forall \theta \in \Theta
    \]
    Thus, $[b]$ satisfies the polar inequalities for all boundary points (and by convexity, for all points in $A$), proving $[b] \in \Delta_C(A)$. Since the inequality is tight at $\theta_0$, $[b]$ lies on the boundary $\partial \Delta_C(A)$.
\end{proof}

If $A$ was assumed \textit{strictly} convex we would have a one-to-one mapping between the polar hyperplane and the point on boundary $\partial \Delta_C(A)$. The Equation \ref{eq:boundary_is_envelope} can thus be modified to say that for a given $[b] \in \mathbb{R}_{n+1}$, $\exists ! \theta \in  \Theta, \ldots $

The polarity plays an important role in projective geometry:
For example, the Pascal's theorem is the polar reciprocal and projective dual of Brianchon's theorem~\cite{richter2011perspectives}.
Polarity is also at the heart of many algorithms in computational geometry~\cite{preparata2012computational,boissonnat1998algorithmic}.

\section{Quadratic polarities, Legendre polarity, and transformations}\label{sec:polartransform}

\subsection{Generalized Legendre-Fenchel transformations}

In~\cite{f1614b9edf1d4404bf7b9a42a7338824}, generalized Legendre-Fenchel transformations (LFTs) of functions $F:\bbR^n \to \barbbR$ are axiomatically characterized as 
 invertible transformations $\calT$ such that
\begin{itemize}
    \item $F_1 \leq F_2 \implies \calT F_2 \leq \calT F_1$,
    \item $\calT F_1 \leq \calT F_2 \implies F_2 \leq F_1$.
\end{itemize}

It is proved that such transformations $\calT$ can be expressed using the ordinary Legendre-Fenchel transformation as
\begin{equation}
(\calT F)(\eta) = \mu F^{*}(E\eta+f)+\langle\eta,g\rangle +h
\end{equation}
where $\mu>0$, $E \in \GL(\bbR^{n}), f,g \in \bbR^{n}$ and $h\in \bbR$.
In plain words, the  general transformations $\calT$ are expressed as affine deformations of  the Legendre-Fenchel transform $F^{*}$.
We can state this result using the function graphs as follows:

\begin{equation}
\begin{bmatrix} 
    \eta \\
    F^{*}(\eta)
\end{bmatrix} \in \graph{F^{*}} \to \begin{bmatrix} 
    E^{-1}(\eta - f) \\
    \mu F^{*} + \langle \eta, E^{-1} g \rangle + h - \langle E^{-1} f,g\rangle 
\end{bmatrix} = \begin{bmatrix} 
    \eta' \\
    \calT F(\eta') 
\end{bmatrix} \in \graph{\calT F)}.
\end{equation}

This affine deformation is can be described in homogeneous coordinates as
\begin{equation}
\begin{bmatrix} 
    \eta \\
    F^{*}(\eta) \\
    1
\end{bmatrix} \in \graph{F^{*}} \to \begin{bmatrix} 
    E^{-1} & 0 & - E^{-1}f \\ (E^{-1}f)^\top & \mu & h - \langle E^{-1} f, g \rangle \\ 0 & 0 & 1
\end{bmatrix} \begin{bmatrix} 
    \eta \\ F^{*}(\eta) \\ 1 
\end{bmatrix} = \begin{bmatrix} 
    \eta' \\
    \mathcal{T}F(\eta') \\
    1
\end{bmatrix} \in \graph{\calT F}.
\end{equation}

\subsection{Generalized LFTs as ordinary LFTs}

In~\cite{nielsen_generalized_2025}, it is shown that those generalized LF transformations can equivalently be described as the ordinary LFT on affinely deformed functions:

\begin{equation}
(\calT F)(\eta)= (\mu F(A\theta + b))^{*}+\langle \theta,c \rangle + d,
\end{equation}
where $A \in \GL(n)$, $b,c \in  \bbR^{n}$ and $d \in \bbR$. 

Thus to get $\calT F$, we can first apply an affine deformation on the graph of $F$ and then compute the LFT on this deformed graph.
The relations between parameters $(\mu, E,f,g,h)$ and $(\mu, A, b,c,d)$ are detailed in \cite{nielsen_generalized_2025}.

\subsection{Quadratic polarities as transformed Legendre polarities}\label{sec:polardeform}

First, let us express an arbitrary quadratic polarity $\Delta_C$ from an ordinary Legendre polarity $\Delta_{\calL}$:

\begin{theorem}[Quadratic polarity as transformed convex body of the Legendre polarity] \label{th:artstein_avidan_generalization_projective}
Let $C \in \GL({n+2})$ and suppose there exists $\calM_T \in \GL({n+2})$ such that:
    $$
    C = C_{\calL} \calM_T^{-1},
    $$
		where $T$ is an affine deformation defined by its matrix $\calM_T = C^{-1}\, C_{\calL}$ as $T:[b] \mapsto \calM_T [b]$.
    Then for any $A \subset \bbR^{n+1}$, we have
    $$
    \Delta_C(A) = T(\Delta_{\calL}(A)).
    $$
\end{theorem}

\begin{proof}
    Let $[b] \in \Delta_C(A)$, and $[b']=\calM_T^{-1}[b]$
    \begin{align*}
        [b] \in  \Delta_{C}(A) &\iff \forall [a] \in A, \quad [a]^{\top} C [b] \geq 0 \\
        & \iff\forall [a] \in A, \quad [a]^{\top} C_{\calL} \calM_T^{-1} [b] \geq 0 \\
        & \iff\forall [a] \in A, \quad [a]^{\top} C_{\calL} \calM_T^{-1} \calM_T [b'] \geq 0 \\
        & \iff\forall [a] \in A, \quad [a]^{\top} C_{\calL} [b'] \geq 0 \\
        & \iff [b'] \in \Delta_{\calL}(A) \\
    \end{align*}
    Since $[b]=T([b'])$ we proved $\Delta_{C}(A)=T(\Delta_{\calL}(A))$
\end{proof}

We can also prove that a quadratic polarity is equivalent to the Legendre polarity on a deformed convex body:

\begin{theorem}[Quadratic polarity as Legendre polarity on deformed convex body] \label{th:nielsen_generalization_projective}
    Let $C \in \GL({n+2})$ and suppose there exists $\calM_S \in \GL({n+2})$ such that:
    $$
    C = \calM_S^{\mathrm{T}} C_{\calL}, 
    $$
		where		$S$ is an affine deformation defined by its matrix $\calM_S =  C_{\calL}\, C^{\top}$ as $S:[a] \mapsto \calM_S [a]$.
				Then for any $A \subset \bbR^{n+1}$:
    $$
    \Delta_C(A) = \Delta_{\mathcal{L}}(S(A)).
    $$
\end{theorem}

The relationships between the two affine deformations $T$ and $S$ are characterized by the following property:

\begin{property}[Relationships between $T$ and $S$]\label{prop:relationsTS}
    The matrices $\calM_{T}$ and $\calM_{S}$ from Theorems \ref{th:artstein_avidan_generalization_projective}  and \ref{th:nielsen_generalization_projective} satisfy:
    $$
    \calM_{T} = C_{\calL}\, \calM_{S}^{-\top}\, C_{\calL}.
    $$
    and
    \begin{equation}
        \calM_{S} = C_{\calL}\, \calM_{T}^{-\top}\, C_{\calL}.
    \end{equation}
\end{property}

\begin{proof} Because $C_{\mathcal{L}}^{-1}=C_{\mathcal{L}}^{\top}=C_{\mathcal{L}}$
   \begin{align*}
    \mathcal{M}_{S} &= C_{\mathcal{L}} C^{\top} \\
    \mathcal{M}_{S}^{-\top} &= C_{\mathcal{L}} C^{-1} \\
    C_{\mathcal{L}} \mathcal{M}_{S}^{-\top} C_{\mathcal{L}} &= C^{-1} C_{\mathcal{L}} \\
    &= \mathcal{M}_{T} \\
   \end{align*}
   Equivalently we can show that $\mathcal{M}_{S} = C_{\mathcal{L}}\mathcal{M}_{T}^{-\top} C_{\mathcal{L}}$
\end{proof}

\section{Polar Fenchel-Young divergences}\label{sec:polarfy}

The Fenchel-Young divergence~\cite{acharyya2013learning} induced by a Legendre-type function $F: \Theta\to\barbbR$ between two dual parameters $\theta_1\in\Theta$ and $\eta_2\in H$ measures the gap of the Fenchel-Young inequality:
$$
Y_F(\theta_1:\eta_2) \eqdef F(\theta_1)+F^*(\eta_2)-\inner{\theta_1}{\eta_2}.
$$

We have $Y_F(\theta_1:\eta_2)\geq 0$ with equality if and only if $\eta_2=\nabla F(\theta_1)$ or equivalently that $\theta_1=\nabla F^*(\eta_2)$.
In information geometry, the Fenchel-Young divergence is a canonical divergence of dually flat spaces~\cite{amari_information_2016}.
The Fenchel-Young divergence finds numerous applications in machine learning (e.g., see~\cite{blondel2020learning,kinoshita2024provable,santos2025hopfield}).
The Fenchel-Young uses the dual mixed parameterization of parameters and equivalently amounts to dual Bregman divergences when rewritten using a single parameterization:
$$
Y_F(\theta_1:\eta_2)=B_F(\theta_1:\theta_2)=B_{F^*}(\eta_2:\eta_1),
$$ 
where a Bregman divergence~\cite{bregman_relaxation_1967} is defined by
$$
B_F(\theta_1:\theta_2)=F(\theta_1)-F(\theta_2)-\inner{\theta_1-\theta_2}{\nabla F(\theta_2)}.
$$

The duality $B_F(\theta_1:\theta_2)=B_{F^*}(\eta_2:\eta_1)$ is called the reference duality in information geometry~\cite{zhang2004divergence}.
A Bregman divergence can be interpreted geometrically as the vertical gap between the points $(\theta_1,F(\theta_1))$ of $\graph{F}$ and 
$(\theta_2, T_{\theta_1}(\theta_2))$ where  $T_{\theta_1}(\theta_2):= F(\theta_2)+\inner{\theta_1-\theta_2}{\nabla F(\theta_2)}$ is the graph of the tangent hyperplane of $\graph{F}$ at $\theta_2$.

Let the natural pairing between a vector $v=(v_1,\ldots,v_n)\in\bbR^n$ and a covector $l^*=(l_1^*,\ldots,l_n^*)\in\bbR_n$ be denoted by $\natural{v}{l^*}\eqdef\sum_{i=1}^n v^i l_i^*$.
The Fenchel-Young divergence can be expressed as
$$
Y_F(\theta_1:\eta_2^*) \eqdef F(\theta_1)+F^*(\eta_2^*)-\natural{\theta_1}{\eta_2^*},
$$ 
where $F^*: H\subset\bbR_n\to\barbbR$.
(In geometric mechanics, the Legendre transform is defined fiberwise using points on the tangent and cotangent bundles~\cite{leok2017connecting}.)

We can define the polar Fenchel-Young divergence using the Legendre polarity as follows:

\begin{definition}[Polar Fenchel-Young divergence]\label{def:breg_div}
Let $A$ be a convex subset of $\bbR^{n+1}$.
 For any pair of points $[a] \in A$ and $[b] \in \Delta(A) \subset \bbR_{n+1}$, we define the polar Fenchel-Young divergence between $[a]$ and $[b]$ as:
\begin{equation}
D_A(a:b) := [a]^\top\, C_{\calL}\, [b],\quad [a] \in A, [b] \in \Delta(A).
\end{equation}
\end{definition}

This definition can be understood geometrically as follows: 
The covector $[b] \in  \bbR_{n+1}$ represent a point in the dual space and can be represented by a polar hyperplane in primal space $\bbR^{n+1}$. 
Using homogeneous coordinates, this polar hyperplane of $[b]$ is given by its normal vector $C_{\calL}\, [b]$. 
The polar Fenchel-Young divergence between two points $[a]$ and $[b]$ is thus the projection of $[a]$ onto this normal vector $C_{\mathcal{L}}\, [b]$.

The polar Fenchel-Young divergences generalize the ordinary Fenchel-Young divergence in the following special case:

\begin{property}\label{prop:polarFYgen}
Let $A = \epi{F}$ where $F$ is a proper closed convex function. 
Let $[a] \in \partial A$ and $[b] \in \partial \Delta_{\calL}(A)$. 
The divergence $D_A(a:b)$ recovers the ordinary Fenchel-Young divergence:
\begin{equation}
D_A(a:b) = F(\theta) + F^*(\eta) - \langle \theta, \eta \rangle = Y_F(\theta:\eta).
\end{equation}
where $[a]$ corresponds to $(\theta, F(\theta))$ and $[b]$ corresponds to $(\eta, F^*(\eta))$.
\end{property}

\begin{proof}
    From Property \ref{prop:legendre_polar_is_legendre_transform} we know that $\Delta_{\mathcal{L}}(\epi{F})=\epi{F^{*}}$. For a closed and proper $F$ and $F^{*}$ we know that $\partial \epi{F}=\graph{F}$ and $\partial \epi{F^{*}}=\graph{F^{*}}$. So we can express $[a]$ and $[b]$ as follows :
for the primal point $[a] \in \partial \epi{F}$, we have $[a] = \begin{bmatrix} 
     \theta & F(\theta) & 1
\end{bmatrix}^\top$. For the dual point $[b] \in \partial \epi{F^*}$, we have $[b] = \begin{bmatrix} 
     \eta & F^*(\eta) & 1
\end{bmatrix} $. 

\begin{align*}
    D_A(a:b) &= [a]^\top C_{\mathcal{L}} [b] \\
    &= \begin{bmatrix} \theta^\top & F(\theta) & 1 \end{bmatrix} 
       \begin{bmatrix} -I_n & 0 & 0 \\ 0 & 0 & 1 \\ 0 & 1 & 0 \end{bmatrix} 
       \begin{bmatrix} \eta \\ F^*(\eta) \\ 1 \end{bmatrix} \\
    &= \begin{bmatrix} \theta^\top & F(\theta) & 1 \end{bmatrix} 
       \begin{bmatrix} -\eta \\ 1 \\ F^*(\eta) \end{bmatrix} \\
    &= -\langle \theta, \eta \rangle + F(\theta) + F^*(\eta)
\end{align*}
This matches the definition of the ordinary Fenchel-Young divergence (or equivalently, its corresponding Bregman divergence).
See Figure~\ref{fig:different_cases_bregdiv} for an illustration.
\end{proof}

\begin{figure}
       \centering
  \def\svgwidth{0.6\columnwidth}%
  \graphicspath{{.}}%
  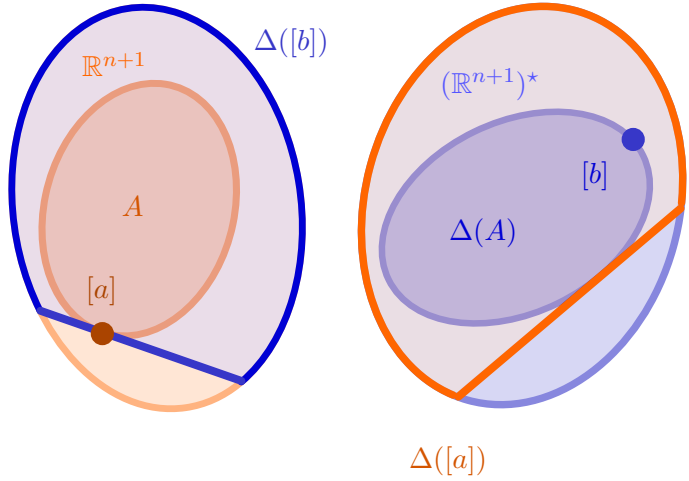%

				\caption{Ordinary Fenchel-Young divergence from the viewpoint of Legendre polarity: $[a] \in \partial A$ and $[b] \in \partial \Delta(A)$.}
        \label{fig:different_cases_bregdiv}
\end{figure}

\begin{property}
The polar Fenchel-Young divergences satisfy the following properties:
\begin{enumerate}
    \item $\forall [a] \in \bbR^{n+1}$ and  
		$[b] \in \Delta_{\calL}(\{a\})$, 
		we have $D_{\{a\}}(a:b) \geq 0$.
		
    \item $D_A(a:b)=0 \iff [b]$ lies on the polar hyperplane $H_{[a]}$ $\iff [a]$ lies on the polar $H_{[b]}$.
		
    \item If $[a] \in \partial A$, then $D_A(a:b)=0 \iff [b] \in \partial\Delta_{\calL}(A)$ (under strict convexity and smoothness assumptions).
\end{enumerate}

\end{property}

\begin{proof}
\begin{enumerate}
    \item By definition of the polar set $\Delta_{\mathcal{L}}(A)$, if $[b] \in \Delta_{\mathcal{L}}(A)$, then for all $[a] \in A$, we have $[a]^\top C_{\mathcal{L}}[b] \geq 0$.
    \item By definition, the polar hyperplane of $[a]$ is the set of points $[x]$ such that $[a]^\top C_{\mathcal{L}} [x] = 0$. Same for polar of $[b]$.
    \item If $[a ]\in \partial A$ and $A$ is strictly convex and smooth, the hyperplane polar to $[a]$ is a supporting hyperplane to the convex set $\Delta_{\mathcal{L}}(A)$. The intersection point of a closed convex set and its supporting hyperplane exists and lies on the boundary. Thus, $[b]$ must be in $\partial \Delta_{\mathcal{L}}(A)$.
\end{enumerate}
\end{proof}

The parameter swapping property of dual Bregman divergences $B_F(\theta_1:\theta_2)=B_{F^*}(\eta_2:\eta_1)$ is expressed in the polar Fenchel-Young divergences as follows:

\begin{property}[Swap of arguments]\label{prop:swap_bregman_arguments}
For any set $A$, $[a] \in A$ and $[b] \in \Delta_{\calL}(A)$, the following equality holds:
\begin{equation}
D_A(a : b) = D_{\Delta_{\calL}(A)}(b : a).
\end{equation}
\end{property}

\begin{proof}
The divergence on the dual set is defined as $D_{\Delta_ \mathcal L(A)}(b:a) = [b]^\top C_{\mathcal{L}}^\top [a]$. Since the Legendre cost matrix $C_{\mathcal{L}}$ is symmetric ($C_{\mathcal{L}}^\top = C_{\mathcal{L}}$):
\begin{align*}
D_A(a : b) &= [a]^\top C_{\mathcal{L}} [b] \\
&= ([a]^\top C_{\mathcal{L}} [b])^\top \quad \text{(scalar transpose)} \\
&= [b]^\top C_{\mathcal{L}}^\top [a] \\
&= [b]^\top C_{\mathcal{L}} [a] \\
&= D_{\Delta_{\mathcal{L}}(A)}(b : a)
\end{align*}
\end{proof}

Figure~\ref{fig:different_cases_bregdiv_cas_in_definition} illustrates a polar Fenchel-Young divergence which extends the classical notion of Fenchel-Young divergence.

    \begin{figure} 
        \centering
  \def\svgwidth{0.6\columnwidth}%
  \graphicspath{{.}}%
  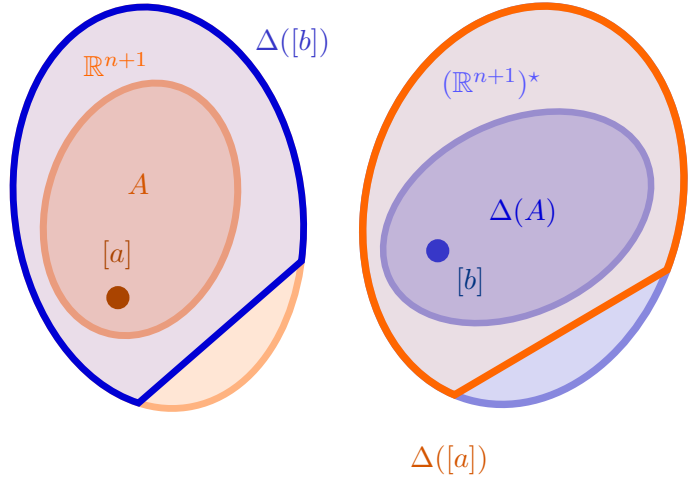%

        \caption{Generalized Fenchel-Young divergence: $[a] \in A$ and $[b] \in \Delta(A)$}
        \label{fig:different_cases_bregdiv_cas_in_definition}
    \end{figure}

\section{Polar total Fenchel-Young divergences}\label{sec:polartbd}

The polar Fenchel-Young divergence between $[a] \in  \bbR^{n+1}$ and $[b] \in  \bbR^{n+1}$ defined above can be seen as the projection of $[a]$ on the vector normal to the polar hyperplane of $[b]$. 
A natural idea would be to normalize this normal vector so that the divergence represent the distance between point $[a]$ and the  hyperplane polar of $[b]$. 
In affine space, this leads to the concept of total Bregman divergences introduced in~\cite{vemuri_total_2011}:
$$
\mathrm{tB}_F(\theta_1:\theta_2)= \frac{1}{1+\inner{\nabla F(\theta_2)}{\nabla F(\theta_2)}} B_F(\theta_1:\theta_2) =\kappa(\theta_2)\, B_F(\theta_1:\theta_2).
$$ 
The factor $\kappa(\theta_)=\frac{1}{1+\inner{\nabla F(\theta)}{\nabla F(\theta)}}$ is called the conformal factor~\cite{wong2018logarithmic} and the total Bregman divergence is a conformal divergence~\cite{nock2015conformal}.

Indeed, if we are interested in the distance in the affine plane $\lambda=1$ between the point $[a]$ and hyperplane of normal vector $C_{\calL} [b]$ we should normalize by $\kappa(b)=\sqrt{1+\left\|\eta_{b}\right\|^{2}}$. 
The factor $\kappa(b)$ corresponds to the norm of the projection in the affine plane of the projective vector $C_{\calL}[b]$. If we decompose $[b]=\begin{bmatrix}
    \eta_b & y_b & \lambda_b
\end{bmatrix}^\top$ it is the norm of the $n+1$ first components of the vector $\begin{bmatrix} 
    -I_n & 0 & 0 \\ 0 & 0 & 1 \\ 0 & 1 & 0 
\end{bmatrix} \begin{bmatrix} 
     \eta_b \\ y_b \\ 1
\end{bmatrix}= \begin{bmatrix} 
    -\eta_b \\ 1 \\ y_b 
\end{bmatrix}.$

Thus a generalization of the polar total Bregman divergence can be defined as follows:

\begin{definition}[Polar total Fenchel-Young divergence]\label{def:totalbreg_div}
$$
\tD_A([a],[b]):= \frac{1}{\kappa(b)} \, D_A([a],[b]).
$$
\end{definition}

 Equivalently for the dual polarity $\Delta^*: \bbR_{n+1}) \to \bbR^{n+1}$,
 we have 
$$
\tD_B([b],[a]) = \frac{1}{\kappa^*(a)} \, D_ B([b],[a]),
$$ 
with $\kappa^*(a)= \sqrt{1+\|\theta_a\|^2}$ the affine norm of vector $C_{\calL}^\top [a]$.

\begin{theorem}[Swap of parameters for polar total Bregman divergences]\label{thm:swaptbD}
    \begin{equation}
        \frac{1}{\kappa^*(a)} \, \tD_A([a]:[b]) = \frac{1}{\kappa (b)}, \tD_{\Delta_{\calL}(A)}([b]:[a]).
    \end{equation}
\end{theorem}

\begin{proof}
\begin{align}
\frac{1}{\kappa^*(a)} \tD_A([a]:[b]) &= \frac{1}{\kappa^*(a)} \frac{1}{\kappa(b)} D_A([a]:[b])\\
    &=  \frac{1}{\kappa^*(a)} \frac{1}{\kappa(b)}D_{\Delta_\mathcal L(A)}([b]:[a]) \quad \text{ from Property \ref{prop:swap_bregman_arguments}}\\
    &= \frac{1}{\kappa(b)} \tD_{\Delta_\mathcal L(A)}([b]:[a]).
\end{align}
\end{proof}

\section{Summary and discussion}\label{sec:concl}

In this work, we have defined the Legendre polarity $\Delta_{\mathcal{L}}$ as the special quadratic polarity which allows one to reframe the Legendre-Fenchel transform.
The Legendre polarity is characterized by the fact that $\Delta_{\mathcal{L}}(\epi{Q})=\epi{Q}$ where $Q$ denotes the paraboloid of equation $Q(\theta)=\frac{1}{2}\sum_i \theta_i^2$. We have shown that generic quadratic polarities can be expressed from the canonical Legendre polarity either by deforming the input set or by deforming the Legendre polar. We then defined the Fenchel-Young polar divergence which extends the Fenchel-Young divergence~\cite{acharyya2013learning} and recover its key properties of non-negativeness and reference duality. We also showed that a natural normalization procedure in the polar Fenchel-Young divergence yields the total Fenchel-Young divergence which is equivalent to the total Bregman divergences~\cite{vemuri_total_2011}.

In optimal transport theory~\cite{villani2009optimal}, a dual formulation of the optimization problem involves the so-called $c$-transform~\cite{wong2018logarithmic} (where ``c'' stands for coupling):
\begin{equation}
\forall \eta \in  H, \quad F^{c}(\eta):= \inf_{\theta \in \Theta}\left(c(\theta,\eta) + F(\theta)\right).
\end{equation}

When $c(\theta,\eta)=-\langle \theta , \eta \rangle$, the $c$-transform becomes the Legendre-Fenchel transform. 

When $c$ is quadratic (i.e., $c(\theta,\eta)=\theta^{2}+ d\eta^{2}+e \theta^\top \eta + f \theta + g \eta + h$ with $(d,e,f,g,h) \in  \mathbb{R}^{5}$), we can describe $c$ by the top left block $C_{n} \in \mathbb{R}^{n}$ of the polarity cost matrix $C$:
\begin{equation}\label{eq:matrix_C_graph_form}
C= \left[\begin{array}{c c c} 
     C_n & 0 & 0 \\ 
     0 & 0 & 1 \\ 
     0 & 1 & 0
\end{array} \right]
\end{equation}

The $c$-transform is then described by the corresponding quadratic polarity $\Delta_C$.

%

\appendix

\section{Proof of involution of the polarity}\label{sec:appendix}

\begin{property}[Polarity involution]
Let $\Delta$ be the polarity associated with a nondegenerate matrix 
$C \in \GL(n+2)$. Then for any closed convex set
$A \subset \bbR^{n+1}$ we have
$$
\Delta_C^*(\Delta_C(A)) = A.
$$
\end{property}

\begin{proof}
$A \subset \Delta^*(\Delta(A))$  : Let $[a]\in A$ and $[b]\in \Delta(A)$.  
By definition of $\Delta(A)$,
\(
[a]^\top C[b]\ge 0.
\)
Since $[b]^\top C^\top[a]=[a]^\top C[b]$, we obtain
\(
[b]^\top C^\top[a]\ge 0.
\)
Hence $[a]\in \Delta^*(\Delta(A))$.

Let prove $\Delta^*(\Delta(A)) \subset A$.
Assume $[a]\notin A$. Because $A$ is a closed convex set and $[a] \notin A$, the hyperplane separation theorem guarantees the existence of a linear functional that strictly separates them. Then, because $C$ is non-degenerate, this functional can be represented as $[b]^\top C^\top$.
\[
[a]^\top C[b] < 0,
\qquad
[a']^\top C[b] \ge 0
\ \text{for all } [a']\in A .
\]
geometrically it means we can find a hyperplane separating $[a]$ and $A$. This works because the set is convex (so supporting hyperplane exists) and closed (so $[a]$ cannot be on the boundary of $A$). From the second condition we know $[b]\in \Delta(A)$, thus from first condition
\(
[b]^\top C^\top[a]<0,
\)
so $[a]\notin \Delta^*(\Delta(A))$.

Therefore $\Delta^*(\Delta(A))=A$.
\end{proof}

\section{Parabola to circle polarity}\label{sec:pcpolar}

Let us consider the polarity which transforms the paraboloid $Q$ into the unit sphere 
$\{(\theta,y) \in  \mathbb{R}^{n+1} \mid \theta_1^{2} + \cdots + \theta_n^{2} +y^{2}=1\}$

\begin{property}    
    [Polarity transforming parabola into circle]\label{prop:parabola2circle}
    If we define a polarity $\Delta_C$ via $C= \begin{bmatrix} 
        I_n & 0 & 0 \\
        0 & \frac{1}{\sqrt{2}} & \frac{1}{\sqrt{2}} \\
        0 & -\frac{1}{\sqrt{2}} & \frac{1}{\sqrt{2}} 
    \end{bmatrix} $ then for $\graph Q = \{(\theta, \frac{\theta^{2}}{2}), \theta \in \Theta \}$
    \begin{equation*}
      \Delta_C(\epi{Q}) = \left\{ (x,y) \in \mathbb{R}^{n}\times \mathbb{R} \mid x^{2}+y^{2} \leq 1 \right\} 
    \end{equation*}
    In other words this polarity transforms the epigraph of the canonical paraboloid into the unit hypersphere.
\end{property}  

\begin{proof}
    We will prove that the image of the epigraph of the parabola $Q : \theta \to \frac{\theta^{2}}{2}$ via polarity $\Delta_C$ gives the unit hypersphere $S := \left\{ \begin{bmatrix}
        x & y & \lambda
    \end{bmatrix}^\top \in \mathbb{R}^{n+1} \mid \left\| x \right\|^{2}+ y^{2} \leq  \lambda^{2} \right\} $. Let $\epi{Q}=\left\{ \begin{bmatrix}
        x & y & \lambda
    \end{bmatrix}^\top \in \mathbb{R}^{n+1} \mid \forall (y,\lambda) \in  \mathbb{R}_+^{2}, \left\|x\right\|^{2} \le 2 y \lambda \right\} $. By definition of the polarity $\Delta_C$ :
    \begin{align*}
      \Delta_C(\epi{Q}) = \left\{ {b} \in  \mathbb{R}^{n+1} \mid \forall [a] \in \epi{Q}, [a]^\top (C [b]) \geq 0 \right\} 
    \end{align*}
    Because $[a]^\top (C [b]) \geq 0$ for all $[a]$ in the cone $K = \epi{Q}$, the vector $C[b]$ must belong to the dual cone $K^*$. 
		Since the cone is self-dual under the standard inner product ($K = K^*$), it follows that $C[b] \in \epi{Q}$. Thus we have :
    \begin{align*}
      \left\|x_b \right\|^{2} & \leq 2 \left( \frac{ \lambda_b + y_b}{\sqrt{2}} \right) \left( \frac{\lambda_b - y_b  }{\sqrt{2}} \right) \\
      & \leq 2 \frac{\lambda_b ^{2} - y_b ^{2}}{2} \\
      & \leq \lambda^{2} - y_b ^{2}
    \end{align*}
    Rearranging gives
    $$
    \left\| x_b \right\|^{2} + y_b ^{2} \leq \lambda ^{2}
    $$
    This is exactly the definition of the hypersphere $S$.
\end{proof}

\bibliographystyle{plain}   
\bibliography{PolarityBIB}

\end{document}